\documentclass[journal]{ieeeconf}  
\usepackage[utf8]{inputenc}
\usepackage{tikz}
\usetikzlibrary{arrows}
 \usetikzlibrary{plotmarks}
 \usetikzlibrary{fit}
 \usetikzlibrary{calc}
 \usetikzlibrary{shapes,positioning}
\usepackage{hyperref}
\usepackage{amsmath,amsfonts}

\usepackage{caption}
\usepackage{subcaption}
\newtheorem{theorem}{Theorem}
\newtheorem{corollary}{Corollary}[theorem]

\newtheorem{remark}{Remark}
\newtheorem{example}{Example}
\newtheorem{assumption}{Assumption}

\newtheorem{definition}{Definition}

\usepackage{amsmath, amssymb}
\usepackage{amsfonts}
\usepackage{graphicx}
\usepackage{enumerate}
\usepackage{ifthen}
\usepackage{float}
\usepackage{cite}
\usepackage{fancyhdr}
\usepackage{pgfplots}
\pgfplotsset{compat=newest}
\usetikzlibrary{patterns}
\usetikzlibrary{decorations.text}
\usepgfplotslibrary{fillbetween}

\usepackage{multicol}
\usepackage{xcolor}

\newboolean{showcomments}
\setboolean{showcomments}{true}

\newcommand{\todo}[1]{  \ifthenelse{\boolean{showcomments}}
{\textcolor{green}{TO DO:  #1}}{}}
\newcommand{\lorenzo}[1]{\ifthenelse{\boolean{showcomments}}
{\textcolor{yellow}{(Lorenzo says: #1)}}{}}
\newcommand{\alain}[1]{\ifthenelse{\boolean{showcomments}}
{\textcolor{red}{(Alain says: #1)}}{}}
\newcommand{\emma}[1]{\ifthenelse{\boolean{showcomments}}
{\textcolor{blue}{(Emma says: #1)}}{}}
\usepackage{blindtext,subcaption}

\IEEEoverridecommandlockouts                              

\renewcommand{\H}{\mathcal{H}}
\newcommand{\A}{\mathcal{A}}
\newcommand{\Y}{\mathcal{Y}}

\title{\LARGE \bf
Population games on dynamic community networks}

\author{Alain Govaert, Lorenzo Zino and Emma Tegling
\thanks{A. Govaert and E. Tegling are with the Department of Automatic Control, Lund University,
        , SE-221 00  Lund, Sweden
        (email: {\tt\small \{alain.govaert,emma.tegling\}@control.lth.se}). L. Zino is with Faculty of Science and Engineering, University of Groningen, 9747 AG Groningen, The Netherlands
        (email: {\tt\small lorenzo.zino@rug.nl}). This work was partially supported by the Wallenberg AI, Autonomous Systems and Software Program (WASP) funded by the Knut and Alice Wallenberg Foundation.}%
}

\begin{document}

\maketitle
\thispagestyle{empty}
\pagestyle{empty}

\begin{abstract}
In this letter, we deal with evolutionary game-theoretic learning processes for population games on networks with dynamically evolving communities. Specifically, we propose a novel mathematical framework in which a deterministic, continuous-time replicator equation on a community network is coupled with a closed dynamic flow process between communities that is governed by an environmental feedback mechanism, resulting in co-evolutionary dynamics. Through a rigorous analysis of the system of differential equations obtained, we characterize the equilibria of the coupled dynamical system. Moreover, for a class of population games ---matrix games--- a Lyapunov argument is employed to establish an evolutionary folk theorem that guarantees convergence to the evolutionary stable states of the game. Numerical simulations are provided to illustrate and corroborate our findings.

\end{abstract}



\section{Introduction}

Evolutionary game theory has rapidly emerged as a powerful mathematical paradigm to model how players in a population game revise their actions to improve their payoff by means of learning mechanisms~\cite{ Weibull1995, Hofbauer2009, Sandholm2010, Quijano2017}. Evolutionary game theory has been successfully adopted to capture many real-world phenomena, including the evolution of behaviors in social, economic, and biological systems~\cite{Traulsen2010,Berg2015}.
In such scenarios, full information on the possible actions and corresponding payoffs is often lacking and players must learn about other actions and payoffs by means of (pairwise) interactions and imitation mechanisms. 
A particular, but important, case is the replicator equation, for which global stability have been established for many classes of games, including stable games~\cite{Hofbauer2009,Fox2013}, potential games~\cite{Sandholm2010}, and matrix games~\cite{Bomze2002,Cressman2014,Riehl2018}. More general forms of imitation dynamics have been considered in~\cite{Hofbauer2009,Barreiro-Gomez2018}, and global convergence results have been established for specific classes of games, including games with strategic substitutes and  strategic complements~\cite{Cimini2017}, some public goods games~\cite{9376279}, and potential games~\cite{cdc2017}.

The literature on evolutionary game theory usually relies on the assumption that individuals interact on a homogeneous time-invariant all-to-all communication structure. However, this assumption is quite simplistic in many real-world scenarios~\cite{Easley2010}. To address this limitation, in particular for the class of learning mechanisms regulated by pairwise interactions and imitation dynamics, some recent efforts toward incorporating a mesoscopic network structure into learning protocols have been made. In these frameworks, players are divided into communities, which determine their possible interactions with other players~\cite{Hofbauer2009,Sandholm2010,Barreiro-Gomez2016,Como2021}. 

In the aforementioned works, 
it is assumed that the communities are fixed a priori~\cite{Hofbauer2009,Sandholm2010,Como2021} or determined by the individuals' actions~\cite{Barreiro-Gomez2016}. 
This relies on a time-scale separation assumption, under which a dynamical co-evolution of the communities at the same time-scale as the learning process is neglected. However, co-evolution is present in many real-world applications, such as ecological and biological systems, in which the presence of environmental feedback can impact the population density in different spatially distributed communities. Another example is socio-economic systems, where seasonality and other cyclic changes may cause migration between geographic regions.

Here, we address this gap concerning co-evolution by proposing a novel dynamic coupling of two mechanisms: an evolutionary dynamics on community networks with closed dynamic flow process as in~\cite{kelly1979reversibility}, and an environmental feedback, which has previously been modeled for evolutionary game frameworks \emph{without} community structure in~\cite{tilman2020evolutionary}. 
We model a scenario where individuals of a community can move to other communities in response to environmental changes. Hence, we augment the system of ordinary differential equations that characterizes the replicator equation on community networks from~\cite{Como2021} 
with a set of ordinary differential equations that describe the evolution of the community densities. 

Our novel modeling approach lays the foundation for a theory to study realistic learning dynamics on networks with dynamic communities and characterize their asymptotic behavior at both the population and community level, without relying on any time-scale separation arguments. Using the proposed framework, we characterize a number of solutions of the dynamical system. For example, we show that if the  proportion of players playing each action converges, then the game reaches a Nash equilibrium, but the community densities may still oscillate.
Moreover, for the special case of binary action sets, we establish an evolutionary folk theorem, proving (almost) global convergence to the evolutionary stable states. These preliminary results pave the way for several extensions, toward further reducing the gap between theoretical and empirical research in population dynamics.

The rest of the letter is organized as follows. In Section~\ref{sec:model}, we formalize the model. In Section~\ref{sec:results}, we present our main convergence results. In Section \ref{section: case study}, we present two case studies. Section~\ref{sec:conclusion} concludes the letter outlining potential avenues for future research. 

\textit{Notation: }The sets of real and non-negative real numbers are denoted by $\mathbb{R}$ and $\mathbb{R}_+$, respectively. For finite sets $\mathcal{A}$ and~$\mathcal{B}$, $\mathbb{R}^{\A\times\mathcal{B}}$ denotes the set of real matrices whose entries are indexed by the elements of $\A\times\mathcal{B}$. 
The transpose of a matrix~$\boldsymbol{x}$ is denoted by $\boldsymbol{x}^\top$. 
The $j$-th column (row) of a matrix~$\boldsymbol{x}$ is denoted by~$\boldsymbol{x}_{j}$ ($\boldsymbol{x}_{j'}$) and the $ij$-th element by $x_{ij}$.
The $i$-th element of a vector $\boldsymbol{y}$ is denoted by ${y}_i$ and the $2$-norm of the vector by $||\boldsymbol{y}||$
The vector of all ones is denoted by $\boldsymbol{1}$ and sign function is denoted by $\mathrm{sgn}$. 
For a non-negative matrix $W$ in $\mathbb{R}^{n\times n}$ the graph associated to~$W$ is defined as~$(\mathcal{N},\mathcal{E}_W)$, with node set~$\mathcal{N}:=\{1,2,\dots,n\}$ and edge set~$\mathcal{E}_W:=\{(i,j)\in\mathcal{N}\times\mathcal{N}: W_{ij}>0\}.$

\section{Model}\label{sec:model}

We consider a continuum of individuals structured into communities that interact with each other through frequency-dependent, instantaneous, random, pairwise encounters with varying strengths, both within and between the communities in the population.
In each pairwise encounter, individuals use an \emph{action} from a finite and common action set $\mathcal A$, which, together with the action of the opponent, results in a reward. 
Evolutionary dynamics describe how the frequencies of actions across the population change under the influence of the pairwise encounters. 
The novel aspect here is that individuals can move freely between the communities. The communities are connected by a dynamic flow network whose flow rates change in response to the frequencies of actions in the communities and, possibly, an exogenous process. 
Since the movement of individuals changes the rate at which pairwise encounters occur, a feedback process is established that describes the co-evolution of strategic interaction and migration flows at the community and population level.
We next provide formal definitions of the various concepts.

\subsection{Population game}

The \textit{population state} $\boldsymbol{y}$ is a vector in the unitary simplex over the action set \(\Y:=\{\boldsymbol{y}\in\mathbb{R}_{+}^\A:\boldsymbol{1}^\top\boldsymbol{y}=1\}.\)
The element $y_i$ of the population state $\boldsymbol{y}$ denotes the fraction of players in the population that use action $i$ in $\A$ ($i$-players). 
A population state $\boldsymbol{y}$ in $\Y$ is said to support action $i$ in $\A$ if a non-zero fraction of the population uses it. The set 
$\mathcal{S}_{\boldsymbol{y}}:=\{i\in\A: y_i>0\}$
is called the \emph{support} of $\boldsymbol{y}$.
Given a population state $\boldsymbol{y}$, expected rewards $r_i(\boldsymbol{y})$ are determined by the reward functions $r_i:\mathcal{Y}\rightarrow\mathbb{R}$, $i$ in $\mathcal A$.
We refer to the pair $(\Y,r)$ as a \emph{population game}.

\subsection{Community network}

Individuals are structured into a finite set $\H$ of \emph{communities}. We refer to the proportion of the population in community $h$ in $\H$ as the \emph{community density} and denote it by $\eta_h$.
The fraction of $i$-players in community $h$ is denoted by $x_{ih}$ and makes up the elements of the \textit{system state} matrix $\boldsymbol{x}$ in $\mathbb{R}_+^{\A\times \H}$. 
The columns of the system state matrix are referred to as the \textit{community state} vectors $\boldsymbol{x}_h$ in $\mathbb{R}_+^{\A}$ for $h$ in $\H$.
Similarly, the support of a community state is
\( \mathcal{S}_{\boldsymbol{x}_h}:=\{i\in\A: x_{ih}>0\}\), with \(\cup_{h\in\H}\mathcal{S}_{\boldsymbol{x}_h}=\mathcal{S}_{\boldsymbol{y}}.\)
The density of the population is assumed constant and given by
\begin{equation}\label{eq: population density}
    \boldsymbol{\eta1}=\boldsymbol{1}^\top \boldsymbol{x1}=\boldsymbol{1}^\top\boldsymbol{y}=1,
\end{equation}
which results in the set of admissible system states 
\(\mathcal{X}=\{\boldsymbol{x}\in\mathbb{R}^{\A\times\H}:\eqref{eq: population density}\}.\)
The strength of interactions between communities is determined by a constant non-negative matrix $W$ in $\mathbb{R}_+^{\H\times\H}$. 
Together with the fraction of $i$-players, this determines the rate
$x_{ih}W_{hk}x_{jk}\geq 0$, for  $i,j$ in $\A$ and $h,k$ in $\H$,
at which $i$-players in community $h$ meet $j$-players in community $k$ in pairwise encounters.
We refer to the triplet $(\H,W,\boldsymbol{\eta})$ as a \emph{community network}.
Throughout the letter the following assumption is made.

\begin{assumption}\label{as: W}
$W$ is non-negative and irreducible with strictly positive diagonal. That is, the graph $(\H,\mathcal{E}_W)$ 
is connected and has self-loops.
\end{assumption}

\subsection{Evolutionary dynamics}

Although the results in Section~\ref{subsec: DB} can be generalized to a broader class of evolutionary imitation dynamics, here we focus on the replicator equation due to its prominence in evolutionary game theory~\cite{cressman2003evolutionary,Cressman2014} and control applications of population games~\cite{Quijano2017}. 
The replicator equation on a community network is a matrix-valued equation $\boldsymbol{f}(\boldsymbol{x})$ in $\mathbb{R}^{\H\times\H}_+$ whose elements
\begin{equation}\label{eq: replicator dynamic}
    \begin{aligned}
        f_{ih}(\boldsymbol{x})=\eta_h\sum_{k\in\mathcal{H}}x_{ik}W_{hk}r_i(\boldsymbol{y})-x_{ih}\sum_{j\in\mathcal{A}}\sum_{k\in\mathcal{H}}x_{jk}W_{hk}r_j(\boldsymbol{y})
    \end{aligned}
\end{equation}
describe how the proportion of $i$-players in community $h$ changes under the influence of selection (see,~\cite{Como2021}).

Two properties of this dynamic deserve some attention. First, if there is a community $h$ in $\H$ such that $\eta_h=1$, then~\eqref{eq: replicator dynamic} reduces to the more familiar form of the single population, single community replicator equation~\cite{Cressman2014}:
\[f_i(\boldsymbol{x})=x_i\Big(r_i(\boldsymbol{x})-\sum_{j\in\A}x_jr_j(\boldsymbol{x})\Big),\quad x_i=W_{hh}x_{ih}.\]
Second, if there exist at least two communities $h,k$ in $\H$ with $W_{hk}>0$ and $\eta_h>0$, then for any system state $\boldsymbol{x}$ in $\mathcal{X}$ such that $x_{ik}>x_{ih}=0$, the right-hand-side of~\eqref{eq: replicator dynamic} reduces to
\(\eta_h\sum_{k\in\mathcal{H}}x_{ik}W_{hk}r_i(\boldsymbol{y}).\) 
This corresponds to the dynamics at a system state in which a supported action of the population state is not supported by all community state vectors. The following assumption ensures the model remains well-defined for such cases.

\begin{assumption}\label{as: positive rewards}
Reward functions $r_i$ are positive-valued.
\end{assumption}

This assumption is born out of a technical necessity that is not confining. In fact, the restricted Nash equilibria of~\eqref{eq: replicator dynamic} are invariant to the addition of a constant to the reward functions. Hence, rewards functions can always be made positive without changing the set of equilibrium points.

\subsection{Dynamic flow process}

We assume individuals of a community have an intrinsic tendency for movement that is described by a constant non-negative matrix $\Lambda$ in $\mathbb{R}_+^{\H\times \H}$ with elements $\lambda_{hk}$.
Given a system state $\boldsymbol{x}$ in $\mathcal{X}$, these intrinsic tendencies are modulated by a non-negative environmental function
    ${\boldsymbol{\phi}}:\mathcal{X}\rightarrow\mathbb{R}^{\H\times\H}_+$.
The environmental response function may depend on a subset of community state vectors (as we shall discuss in Example~\ref{example one}) or on an exogenous variable. It may also be governed by positive system dynamics (see, Example~\ref{example 2} in the following).
As in the closed migration processes of~\cite[Chapter 2]{kelly1979reversibility}, we assume scaling is multiplicative such that the dynamic rate at which individuals move from community $h$ to community~$k$ is $\lambda_{hk}\phi_{hk}(\boldsymbol{x})$ in $\mathbb{R}_+$. 
The changes in community densities induced by these movements are described by the dynamic flow process:
\begin{equation}\label{eq: eta dynamic}
    \dot{\eta}_h=\underbrace{\sum_{k\in\mathcal{H}}\lambda_{kh}\phi_{kh}(\boldsymbol{x})\eta_k}_{\text{inflow}}-\underbrace{\eta_h\sum_{k\in\mathcal{H}}\lambda_{hk}\phi_{hk}(\boldsymbol{x})}_{\text{outflow}}.
\end{equation}
This preserves the population density~\eqref{eq: population density} because the system is closed. Moreover, non-negativity of the environmental function $\boldsymbol{\phi}$ and movement matrix $\Lambda$ ensures the solutions of~\eqref{eq: eta dynamic} remain well-defined densities in the unit interval.
\begin{remark}
For ecological population models, the existence of a degenerate equilibrium at which $\eta_h=0$ is unrealistic because empty habitats tend to become re-occupied~\cite{cressman2006migration}. To ensure this, $\mathcal{E}_\Lambda$ must be connected and the environmental response $\phi_{kh}(\boldsymbol{x})>0$, for all $\boldsymbol{x}$ in $\mathcal{X}$, if $\eta_h=0$.
\end{remark}
We illustrate dynamic flow processes through the following two examples.

\begin{example}[Rational community selection]\label{example one}
Let us associate to each community $h$ in $\mathcal{H}$ a payoff function $\pi_h(\eta_h)$ in $\mathbb{R}$ that determines the payoff that individuals obtain for being in community $h$ when its density is $\eta_h$. 
The set of communities with the highest payoff in the neighborhood of community $k$ depends on $\boldsymbol{\eta}$ and is denoted by 
\(\mathcal{N}_k(\boldsymbol{\eta}):=\{h\in\mathcal{H}: h\in\mathrm{argmax}_{z:(z,k)\in\mathcal{E}_{\Lambda}}\pi_z(\eta_z)\}.\)
Now, consider the environmental function 
\begin{equation}\label{eq: BR}
    \phi_{kh}(\boldsymbol{\eta})=\begin{cases}
    1/|\mathcal{N}_k(\boldsymbol{\eta})| \quad &\text{\ if\ } h\in\mathcal{N}_k(\boldsymbol{\eta})\,,\\
    0 \quad &\text{\ otherwise}.
    \end{cases}
\end{equation}
If the graph $(\H,\mathcal{E}_\Lambda)$ is complete with self-loops, this dynamic flow process corresponds to an unconstrained best response dynamic in which individuals move only to those communities for which they receive the maximum payoff. Indeed, this corresponds to the habitat selection games studied in~\cite{cressman2006migration, Cressman2014}, which are discussed further in Section~\ref{section: case study}.
On the other hand, if $(\H,\mathcal{E}_\Lambda)$ is not complete, individuals can move only to those neighboring communities with the highest local payoff. 
This results in a constrained best response process where the feasible actions of individuals in a community are determined by its local payoff performance, as in relative best response dynamics in network games~\cite{9029821}.
\end{example}

\begin{example}[Environmental out-migration]\label{example 2}
Density limiting effects are an important consideration in ecological and population models. 
Increased pressures on resources can lower reproductive rates~\cite{CRESSMAN2003519} and increase out-migration~\cite{isard1966methods,Masa}. 
The latter, for example, can be captured by considering a dynamic environmental response that is uniform in the \emph{outflows} of a community: 
\begin{align}\label{eq: phi dot}
\dot{\phi}_{hk}=(\phi_{hk}-m)\left(1-\frac{\eta_h}{\kappa_h}\right)\phi_{hk}\,,
\end{align}
where $m>0$ is a constant maximum environmental response and $\kappa_h>0$ the carrying capacity of community $h$.
Note that the non-negative solution $\phi_h(t)$ of the above differential equation exists and result in well-defined community densities. Moreover, if a community is overcrowded $\eta_h>\kappa_h$, then the out-migration increases.
This example may be generalized to account for dependency of the carrying capacity on the community state vector $x_h$, akin to density games~\cite{NOVAK201326}.
\end{example}

When the dynamic movement rates depend on the system state but are action independent, the proportion of actions in the outflows of a community are distributed according to the corresponding community state vector. 
Thus, when~\eqref{eq: replicator dynamic} is interconnected with~\eqref{eq: eta dynamic}, the closed-loop system state dynamic reads as
\begin{multline}\label{eq: closed}
    \dot{x}_{ih}={\sum_{k\in\mathcal{H}}\lambda_{kh}\phi_{kh}(\boldsymbol{x})x_{ik}}-{x_{ih}\sum_{k\in\mathcal{H}}\lambda_{hk}\phi_h(\boldsymbol{x})}\\+\eta_h\sum_{k\in\mathcal{H}}x_{ik}W_{hk}r_i(\boldsymbol{y})-x_{ih}\sum_{j\in\mathcal{A}}\sum_{k\in\mathcal{H}}x_{jk}W_{hk}r_j(\boldsymbol{y}).
\end{multline}

\begin{definition}
The combination of a population game, a community network and an environmental function, define the \emph{population game on a dynamic community network} as the tuple $\Gamma=(\mathcal{Y},r, \H,W,\Lambda,\boldsymbol{\phi})$.
\end{definition}
The model framework is illustrated in Figure~\ref{fig:networkillustration}.

\begin{figure}
     \centering
      \includegraphics[width = 0.35\textwidth]{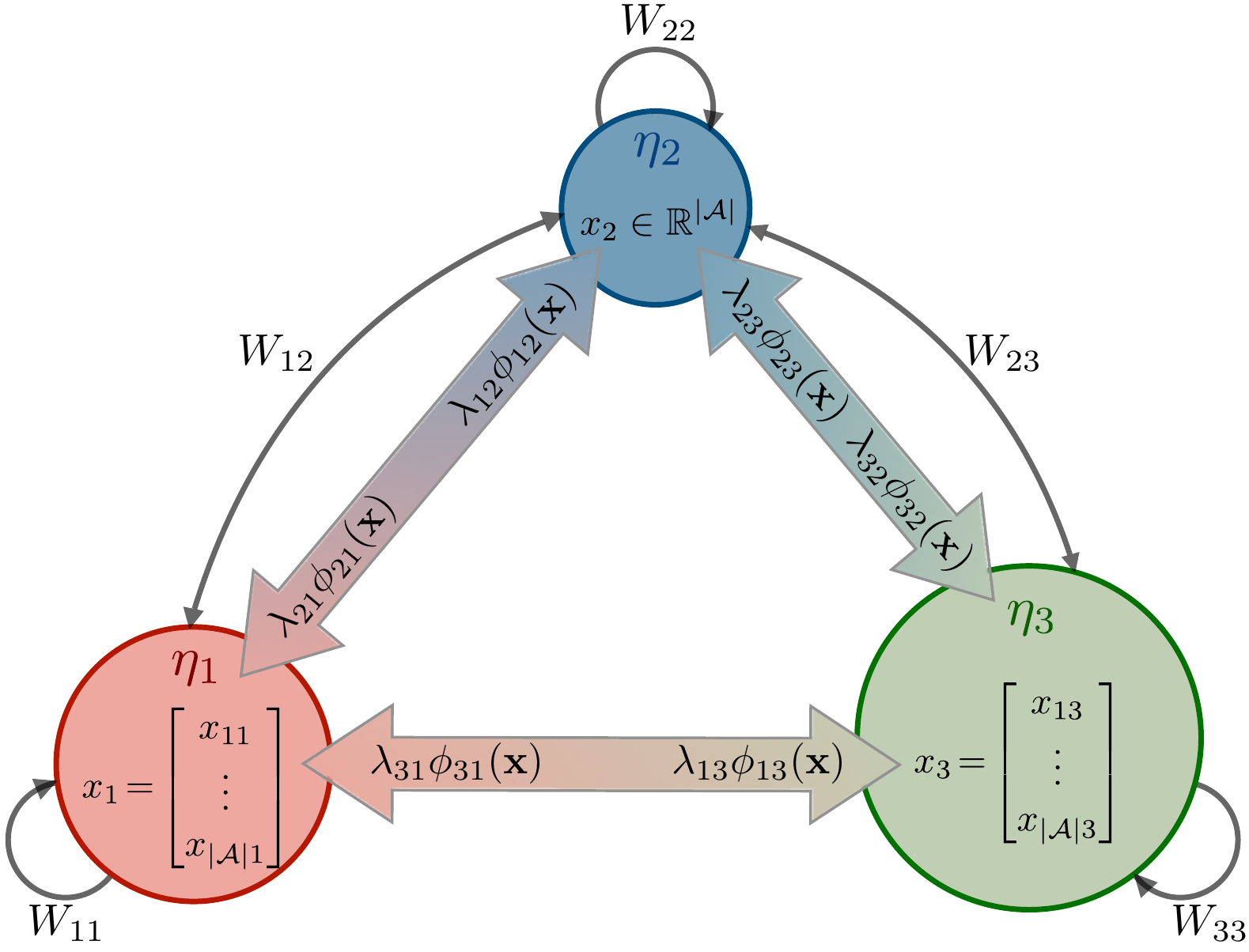}
     \caption{Illustration of the population game over a three-community network with dynamic densities~$\eta_h$.
     We depict a symmetric interaction matrix~$W$ (i.e., $W_{hk} = W_{kh}$), though Theorem~1 does not require such symmetry. The non-negative inter-community movement matrix $\Lambda$ here describes a complete graph as in Example~\ref{example one} that, combined with the environmental response function in Example~\ref{example one}, results in an unconstrained best-response process. In general, there are no connectivity requirements for $(\mathcal{H},\mathcal{E}_\Lambda)$.}
     \label{fig:networkillustration}
 \end{figure}

\section{Results}\label{sec:results}

\subsection{A dynamic system state-density balance}\label{subsec: DB}

This section characterizes the asymptotic relation between the population state, the community state and dynamic community densities of a connected community network. 
Here, no further restrictions are imposed on the size of the finite action set, the structure of the positive reward functions, or the non-negative environmental function.

\begin{theorem}\label{theorem: DD}
 Consider a population game on a dynamic community network $\Gamma$ that satisfies Assumptions~\ref{as: W}-\ref{as: positive rewards}.
Let $\boldsymbol{\eta}(t)$ and $\boldsymbol{x}(t)$ be the solutions of the dynamics~\eqref{eq: eta dynamic} and~\eqref{eq: closed}, respectively. If
	  \(\underset{t\rightarrow\infty}{\mathrm{lim}}\boldsymbol{x}(t)\boldsymbol{1}=\boldsymbol{y}^\ast,\)
	    then $\boldsymbol{y}^\ast$ is a restricted Nash equilibrium and, for all $h$ in $\H$ such that $\underset{t\rightarrow\infty}{\mathrm{lim\ inf}}\ \eta_h(t)>0$, it holds that
\begin{equation}\label{eq: dynamic balance}
    \underset{t\rightarrow\infty}{\mathrm{lim}}\frac{{x_{ih}}(t)}{\eta_h(t)}=y_i^\ast\,, \quad \forall i\in\A.
\end{equation}
In the degenerate cases $\eta_k(t)=0$, ${x_{ik}(t)}=0$ for all $i$ in $\A$.
\end{theorem}

\begin{proof} 
The degenerate cases are trivial since a community with zero density does not contain individuals.
We proceed to prove the remaining statements.
From Assumption~\ref{as: positive rewards} and \cite[Theorem 1]{Como2021} it follows that $f_{ih}({\boldsymbol{x}})$ in \eqref{eq: replicator dynamic} is zero for $\eta_h>0$ if and only if the following two conditions hold:
1) $r_i(\boldsymbol{y})=r_j(\boldsymbol{y})$ for all $i,j\in\mathcal{S}_{\boldsymbol{y}}$; and 2)
${x}_{ih}=y_i\eta_h$ for all $i\in\A$ and $h\in\H$. 
The first condition corresponds to the population state vector $\boldsymbol{y}$ being a restricted Nash equilibrium; the second refers to a balanced system state that holds trivially in the degenerate case $\eta_h=0$.
Since the dynamic flow process is closed, the sum of inflows equals the sum of outflows and
\begin{equation}\label{eq: same population equilibria}
    \sum_{h\in\H}f_{ih}(\boldsymbol{x})=\sum_{h\in\H}\dot{x}_{ih}\,, \quad \forall \boldsymbol{x}\in\mathcal{X}.
\end{equation}
That is,~\eqref{eq: replicator dynamic} and~\eqref{eq: closed} have the same set of equilibrium population state vectors. Hence, $\boldsymbol{y}^\ast$ must also be a restricted Nash equilibrium.
Next, we show that Condition 2 (which holds due to Assumption~\ref{as: positive rewards}) implies that for $\eta_h>0$,
\begin{equation}\label{eq: ratio}
    \frac{d}{dt}\frac{x_{ih}}{\eta_h}
   =\frac{1}{\eta_h^2}\sum_{k\in\mathcal{H}}\lambda_{kh}\phi_{kh}\left(x_{ik}\eta_h-x_{ih}\eta_k\right)+\frac{f_{ih}(\boldsymbol{x})}{\eta_h}=0.
\end{equation}
Clearly, Condition $2$ ensures that the first term in~\eqref{eq: ratio} is zero, while the second term is zero if Condition $1$ holds as well, and thus the equality holds. 

It remains to show that the limit $\boldsymbol{x}(t)\boldsymbol{1}=\boldsymbol{y}^\ast$ for $t\rightarrow\infty$, implies that $f_{ih}=0$ all $i$ in $\A$ and $h$ in $\H$.
Assume, for the sake of contradiction, that $\boldsymbol{x}(t)\boldsymbol{1}=\boldsymbol{y}^\ast$ and $f_{ih}\neq 0$. Then, Conditions 1 and 2 cannot both hold.
In particular, if Condition $1$ is violated then $\boldsymbol{y}^\ast$ is not a restricted Nash equilibrium. 
By \eqref{eq: same population equilibria} this is a contradiction. 
Suppose now, that Condition $1$ holds while Condition $2$ is violated. 
Because Condition 1 holds,
$\sum_{l\in\H}f_{il}=\dot{y}_i=0$ for all $i$ in $\A$. This implies that either $f_{il}(\boldsymbol{x})=0$ for all $i$ in $\A$ and $l$ in $\H$ (and thus Condition 2 holds) \emph{or} there exist $h,k$ in $\H$ such that $f_{ih}>0$ and $f_{ik}<0$. Then, at the population state equilibrium $\boldsymbol{y}^\ast$
\begin{equation}\label{eq: fih}
    f_{ih}(\boldsymbol{x})=\sum_{k\in\mathcal{H}}W_{hk}\left(x_{ik}\eta_h-x_{ih}\eta_k\right)r_i(\boldsymbol{{y}^\ast})>0.
\end{equation}
Dividing by the strictly positive term $r_i(\boldsymbol{y^\ast})\eta_h\sum_{k}W_{hk}\eta_k$,~\eqref{eq: fih} yields the inequality
\[\sum_{k\in\H}Q_{hk}\alpha_k>\alpha_h \text{\ with\ } Q_{hk}:=\frac{W_{hk}\eta_k}{\sum_{l\in\H}W_{hl}\eta_l},\quad \alpha_k:=\frac{x_{ik}}{\eta_k}.\]
Because the matrix $Q$ in $\mathbb{R}^{\H\times \H}_+$ with elements $Q_{hk}$ is irreducible and stochastic, all its eigenvalues are within the unit circle. Consequently, there does not exist a vector $\alpha$ in $\mathbb{R}^{\H}_+$ with elements $x_{ik}/\eta_k$ such that the above inequality is satisfied. At the population state equilibrium $\boldsymbol{y}^\ast$, the inequality \eqref{eq: fih} cannot be true for any $h$ in $\H$ and $i$ in $\A$. 
Hence, the only possibility is then $f_{ih}(\boldsymbol{x})=0$ for all $i$ in $\A$ and $h$ in $\H$, which implies Condition $2$ must hold.
\end{proof}
 
Theorem~\ref{theorem: DD} shows that the dynamic system state-density balance in~\eqref{eq: dynamic balance} is achieved even when the dynamic flow process~\eqref{eq: eta dynamic} is non-convergent, i.e. $\underset{t\rightarrow\infty}{\mathrm{lim}}\boldsymbol{\eta}(t)$ does not exist. An example of such a case is illustrated in Section~\ref{section: case study SM} and Fig.~\ref{fig: fig1}, but many more may be considered. 
An example in which \eqref{eq: eta dynamic} does converge to a density distribution is provided in Section~\ref{section: case study IFD}. In such an event, the following corollary is an immediate consequence of Theorem~\ref{theorem: DD}.
 
\begin{corollary}\label{cor: convergent eta}
If the dynamic flow process \eqref{eq: eta dynamic} is such that $\underset{t\rightarrow\infty}{\mathrm{lim}}\boldsymbol{\eta}(t)=\boldsymbol{\eta}^\ast$ and the population state satisfies $\underset{t\rightarrow\infty}{\mathrm{lim}}\boldsymbol{x}(t)\boldsymbol{1}=\boldsymbol{y}^\ast$, then the system state matrix converges to the equilibrium
\(\underset{t\rightarrow\infty}{\mathrm{lim}}{\boldsymbol{x}}(t)=\boldsymbol{y}^\ast\boldsymbol{\eta}^{\ast\top}.\)
\end{corollary}

\begin{remark}
If $W=W^\top$ and $\phi_h(\boldsymbol{x})=0$ for all $\boldsymbol{x}$ in $\mathcal{X}$ and $h$ in $\H$, then Theorem~\ref{theorem: DD} coincides with the result in \cite[Proposition 3]{Como2021} obtained for a broad class of evolutionary imitation dynamics. The proof of Theorem~\ref{theorem: DD} can be extended to this class of evolutionary imitation dynamics using analogous arguments.
\end{remark}

\subsection{Evolutionary stability}

This section focuses on the relation between the population state vector and evolutionarily stable states. For this result, we restrict our attention to binary action sets and matrix games with rewards functions of the form
\begin{equation}
r(\boldsymbol{y})=\begin{bmatrix}\label{eq: HD reward}
a&b\\c&d
\end{bmatrix}\boldsymbol{y}\,,\quad  a, b, c, d>0,   
\end{equation}
which satisfy Assumption~\ref{as: positive rewards}. These rewards can also be interpreted as the payoffs of a player in a two-by-two symmetric matrix game played against the mixed strategy $\boldsymbol{y}$~\cite{hofbauer_sigmund_1998}. 
Consider the following definition of an evolutionarily stable state.

\begin{definition}[Evolutionarily stable state~\cite{hofbauer_sigmund_1998}]\label{def: Ess}
A population state vector $\hat{\boldsymbol{y}}$ in $\mathcal{Y}$ is an evolutionarily stable state if there exists $\delta>0$ such that
$\boldsymbol{\hat{y}}\cdot A\boldsymbol{y}>\boldsymbol{y}\cdot A\boldsymbol{y}$,
for all $\boldsymbol{y}$ with $0<||\boldsymbol{y}-\hat{\boldsymbol{y}}||<\delta$. 
\end{definition}

\textcolor{black}{The following result shows the importance of evolutionarily stable states also for the replicator equation on networks with dynamic communities when the underlying interaction network is undirected.}

\begin{theorem}\label{theorem: ESS}
Consider $\Gamma$ that satisfies Assumption~\ref{as: W} and additionally assume that $W=W^\top$, the action set is binary and the game is a matrix game, as in \eqref{eq: HD reward}. Then,
\begin{enumerate}
    \item an evolutionarily stable state $\hat{\boldsymbol{y}}$ in $\mathcal{Y}$ is locally asymptotically stable; and
    \item if an evolutionarily stable state $\hat{\boldsymbol{y}}$ exists in the interior of $\mathcal{Y}$ then all interior trajectories converge to it.
\end{enumerate}
\end{theorem}

\begin{proof}
Consider the non-negative function 
    $P(\boldsymbol{y})=\prod_{i\in\mathcal{S}_{\hat{y}}}y_i^{\hat{y}_i}$,
which has a unique maximum at $y_i=\hat{y}_i$ for all $i$ in $\mathcal{A}$ and is a local Lyapunov function for matrix games under the single population, single community,  replicator equation~\cite{hofbauer_sigmund_1998}. 
For all $\boldsymbol{y}$ in $\mathcal{Y}$ such that $y_i>0$ if $\hat{y}_i>0$ it holds that $P(\boldsymbol{y})>0$ and for symmetric $W$ and binary action sets $\mathcal{A}=\{i,j\}$, we write
\begin{equation}
\label{eq: binary Lyapunov}
        \frac{\dot{P}(\boldsymbol{y})}{P(\boldsymbol{y})}=\sum_{i\in\A}\frac{\hat{y}_i}{y_i} \boldsymbol{x}_{i'}W\boldsymbol{x}_{j'}^\top\left( r_i(\boldsymbol{y})-r_j(\boldsymbol{y})\right).
\end{equation}
For $P(\boldsymbol{y})>0$, it also holds that $\mathrm{sgn}(\hat{y}_i/y_i)=\mathrm{sgn}(\hat{y}_i)$. Moreover, since $W$ is connected, non-negative, and has strictly positive diagonal entries, it holds that
\(\mathrm{sgn}(\boldsymbol{x}_{i'}W\boldsymbol{x}_{j'}^\top)=\mathrm{sgn}(y_iy_j).\)
Combined with~\eqref{eq: binary Lyapunov}, it follows that for $P(\boldsymbol{y})>0$ and
\begin{equation}\label{eq: sign}
\mathrm{sgn}\left(\frac{\dot{P}(\boldsymbol{y})}{P(\boldsymbol{y})}\right)=\mathrm{sgn}\left(\sum_{i\in\A}\hat{y}_iy_j\left( r_i(\boldsymbol{y})-r_j(\boldsymbol{y})\right)\right).
\end{equation}
Because $\boldsymbol{1^\top\hat{y}}=1$, the term within the sign function of the right-hand side of~\eqref{eq: sign} can be written as
$
         \sum_{i}\hat{y}_iy_j\left(\left( r_i(\boldsymbol{y})-r_j(\boldsymbol{y})\right)\right)
         =\boldsymbol{\hat{y}}\cdot A\boldsymbol{y}-\boldsymbol{y}\cdot A\boldsymbol{y}\label{eq: ess}.
$
Then, by definition~\ref{def: Ess}, there exists $\delta>0$ such that the right-hand side of \eqref{eq: binary Lyapunov} is strictly positive for all $\boldsymbol{y}:0<||\boldsymbol{y}-\hat{\boldsymbol{y}}||<\delta$, making it a strict local Lyapunov function for $\Gamma$ and the closed-loop replicator equation \eqref{eq: closed}.
Moreover, if an interior evolutionarily stable state exists, it coincides with the unique interior Nash equilibrium of the binary matrix game~\cite{cressman2003evolutionary} and all interior trajectories converge to it.
\end{proof}

\begin{remark}
The combination of Theorem~\ref{theorem: DD} and \ref{theorem: ESS} characterizes the asymptotic behavior of the system at a population and community level and shows the important role of evolutionarily stable states. In this way, it extends the famous evolutionary folk theorem~\cite{cressman2003evolutionary} to population games on (symmetric) dynamic community networks and binary action sets. It is known, however, that the folk theorem does not apply to other important evolutionary dynamics~\cite{Cressman2014} and, for dynamic communities, the extension to more than two actions is still open. To overcome these limitations, one can use the well-established framework of potential games and strictly stable games to show (global) convergence of the population state and obtain the characterization of the relation between population states, community states and dynamic community densities as before.
\end{remark}

\section{Case Studies}\label{section: case study}

\subsection{Ideal free distribution}\label{section: case study IFD}

Here, we illustrate the utility of our results for a three-community network $\mathcal{H}=\{1,2,3\}$ obeying the replicator equation combined with a dynamic flow process from Example~\ref{example one} leading to a coupled dynamic process of rational community selection and evolution of actions. We assume a complete movement matrix such that all the entries of $\Lambda$ are equal to $1$. 
Consider the community payoff function
\begin{equation}\label{eq: payoff com}
    \pi_h(\eta_h)=\alpha_h\left(1-\frac{\eta_h}{\kappa_h}\right),\quad \  \kappa_h,\alpha_h>0,~  \forall h\in\H\,,
\end{equation}
proposed in~\cite{cressman2006migration}. If the community density $\eta_h$ is below (above) its carrying capacity $\kappa_h$, individuals obtain a positive (negative) payoff scaled by the constant $\alpha_h$. 
By the definition of the environmental function~\eqref{eq: BR}, it follows that $\boldsymbol{\phi}(\boldsymbol{\eta})$ is row-stochastic. Consequently, the dynamic flow process~\eqref{eq: eta dynamic} becomes a piecewise linear system that can be written in vector form as
\begin{equation}\label{eq: piece}
    \dot{\boldsymbol{\eta}}=\boldsymbol{\phi}(\boldsymbol{\eta})^\top\boldsymbol{\eta}-\boldsymbol{\eta}.
\end{equation}
Through~\eqref{eq: closed}, this dynamic flow process can be interconnected with a binary actions set and symmetric pairwise rewards~\eqref{eq: HD reward} with $0<a<c$ and $0<d<b$ corresponding to the Hawk-Dove game~\cite{cressman2003evolutionary}. 
If the interaction matrix $W$ is symmetric, it follows from Theorem~\ref{theorem: ESS} that for every initial population state $\boldsymbol{y}(0)$ in the interior of $\mathcal{Y}$, the solution of~\eqref{eq: closed} converges to the unique interior evolutionarily stable state 
\begin{equation}\label{eq: interior ESS}
\hat{\boldsymbol{y}}^\top=\begin{bmatrix}\frac{b-d}{c-a+b-d}&\frac{a-c}{a-c+d-b}
\end{bmatrix}.
\end{equation}
Since the payoff function~\eqref{eq: payoff com} is decreasing in $\eta_h$ and some individuals always move to the community with the highest payoff due to~\eqref{eq: BR}, the solution of the dynamic flow process~\eqref{eq: piece} converges to the unique \emph{ideal free distribution}~\cite[Theorem 3]{Cressman2014} where the payoffs~\eqref{eq: payoff com} are equal at all communities. 
If the sum of the carrying capacities is equal to the population density~\eqref{eq: population density}, this simplifies to a \emph{balanced dispersal}~\cite{w2004dispersal}  
at which $\eta^\ast_h=\kappa_h$ for all $h$ in~$\H$.
It is useful to note that this process can be extended to an arbitrary---albeit finite---number of communities on a complete graph.
The influence of a less connected graph on the equilibrium distribution of $\boldsymbol{\eta}$ are not yet known, but make out an important aspect of real-world habitat selection processes~\cite{gustafson1998quantifying}.  
Corollary~\ref{cor: convergent eta} of Theorem~\ref{theorem: DD} ensures the system state $\boldsymbol{x}$ converges asymptotically to the equilibrium $\boldsymbol{x}^\ast=\hat{\boldsymbol{y}}\boldsymbol{\eta}^{\ast\top}$ at which the community state vectors are proportional to the product of the evolutionarily stable state and the ideal free distribution.
\begin{figure*}
    \centering
    \begin{subfigure}[t]{0.49\textwidth}
    \centering
    \includegraphics[width=\textwidth]{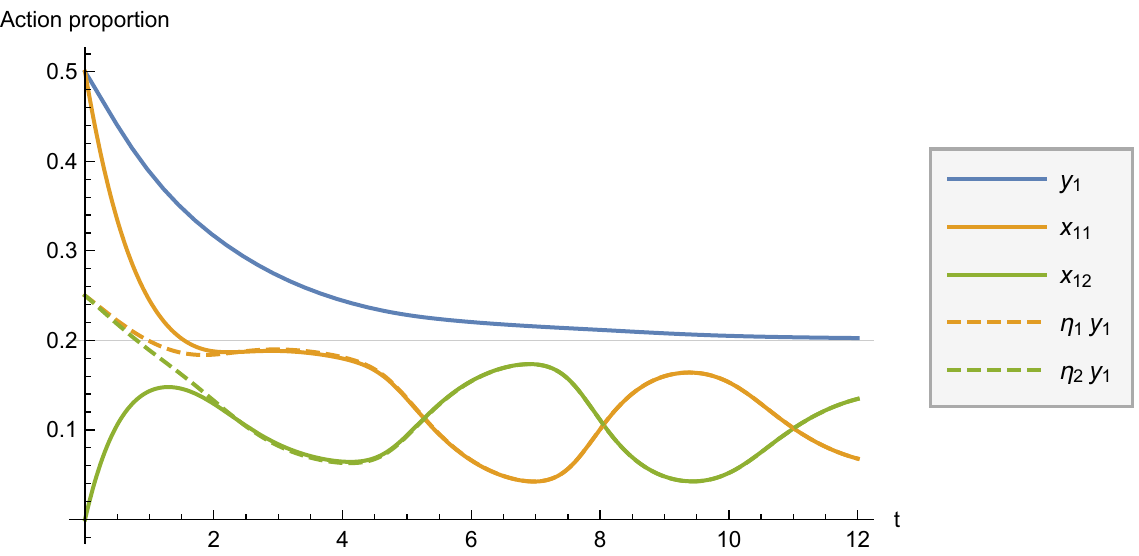}
         \caption{}\label{fig: a}
    \end{subfigure}
\hfill
    \begin{subfigure}[t]{0.49\textwidth}
    \centering
    \includegraphics[width=\textwidth]{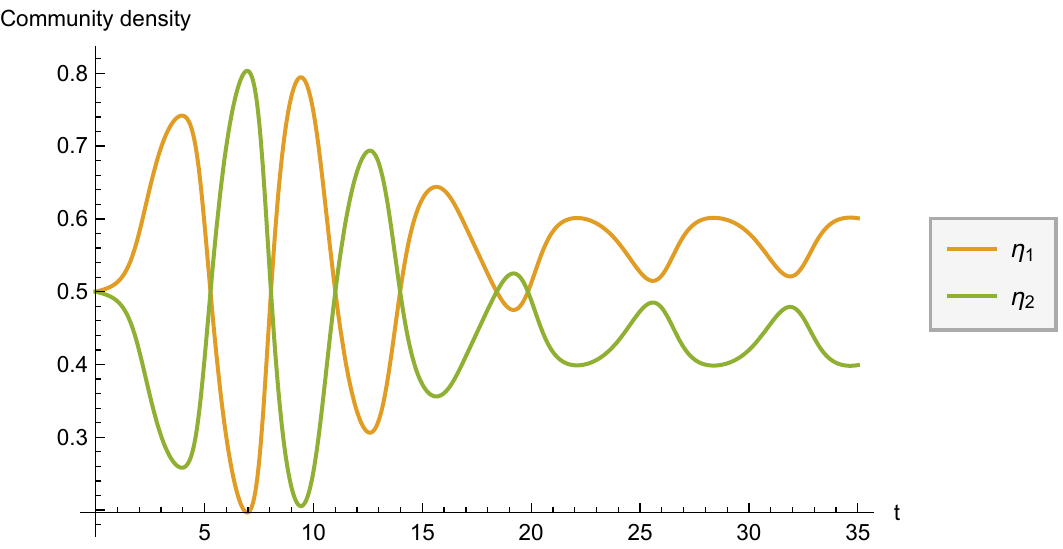}
        \caption{}\label{fig: b}
    \end{subfigure}
    \vspace{7pt}
    \caption{Numerical solutions of the dynamic flow process~\eqref{eq: eta dynamic} and the closed-loop system~\eqref{eq: closed} for two communities with the dynamic environmental function~\eqref{eq: phi dot} with sinusoidal varying carrying capacities~\eqref{eq: sinus}. 
    Notice in (b) that the community densities oscillate periodically from $t=22$, while the population state in (a) converges to an equilibrium and a dynamic balance system state is achieved asymptotically. 
   Parameter values: $\gamma=0.25$, $\rho=0.5$, $a=1$, $b=7$, $c=5$, $d=6$, $W_{11}=0.7$, $W_{12}=W_{21}=W_{22}=0.3$, $\lambda_{aa}=\lambda_{bb}=1$, $\lambda_{ba}=0.8$, $\lambda_{ab}=0.5$ with $\phi_{12}(0)=\phi_{21}(0)=0.05$.}\label{fig: fig1}
\end{figure*}

\subsection{Seasonal migration}\label{section: case study SM}

The ideal-free distribution case study is an example in which both the evolutionary and dynamic flow process are convergent. Theorem~\ref{theorem: DD}, however, only requires the population state vector to converge to an equilibrium ---not the system state and community densities. In fact, in reality, they often exhibit periodic behaviors under the influence of geographically distributed seasonal changes or socio-economic cyclic phases.  
We illustrate this for two communities with a sinusoidally varying carrying capacity, from~\cite{banks1993growth}:
\begin{equation}\label{eq: sinus}
    \kappa_1(t)=\gamma\ \mathrm{sin}(t)+\rho\,,\quad
    \kappa_2(t)=\gamma\ \mathrm{sin}(t+\pi)+\rho\,,
\end{equation}
with $0<\gamma<\rho$ to ensure they are positive. 
The phase shift between the two communities represents a geographical difference in seasonal changes or socio-economic parameters.
The sinusoidal varying carrying capacities is combined with the dynamic environmental function~\eqref{eq: phi dot}.
A closed system of ODEs is obtained with the dynamic flow process~\eqref{eq: eta dynamic} and the evolutionary dynamic~\eqref{eq: closed} with rewards~\eqref{eq: HD reward}.

Even for just two communities, a full analysis is challenging. 
However, with the theory developed in this letter, some critical insights at both the population and community level can be obtained. 
As Theorem~\ref{theorem: ESS} predicts, the population state converges asymptotically to the evolutionarily stable state \eqref{eq: interior ESS} indicated by the horizontal line in Fig.~\ref{fig: a}. 
A non-trivial behavior occurs at a community level, whereby persistent oscillations due to the sinusoidally varying carrying capacities emerge, as reported in Fig.~\ref{fig: b}. 
The effect of Theorem~\ref{theorem: DD} is then seen by the trajectories that converge to each other in Fig.~\ref{fig: a}: the proportion of players in the communities converges asymptotically to the product of the oscillating community densities and the evolutionarily stable state.

\section{Conclusion}\label{sec:conclusion}

In this letter, we have proposed a novel mathematical framework for evolutionary game-theoretic learning processes on dynamic community networks. Specifically, our framework couples a replicator equation on a community network with a closed dynamic flow process and an environmental feedback, co-evolving at comparable time scales. 
Under reasonable assumptions on the structure of the networks and reward functions, we have provided a characterization of the equilibria of the dynamical system. Moreover, for matrix games on symmetric networks, we have established convergence to the evolutionary stable states.

Our framework paves the way for promising avenues of future research. First, while our analysis focuses on closed flow processes, the framework could be expanded to deal with open flow processes and to incorporate noise and uncertainties. 
Second, the two processes can be further intertwined by considering density-dependent rewards, as well as action-dependent migration rates. 
Third, the stability of equilibrium points should be further investigated toward fully characterizing the asymptotic behavior of the system, with specific focus on the role of symmetry of the interaction and movement matrices.








\bibliographystyle{IEEEtran}
\bibliography{bib,IEEEabrv}


\end{document}